\documentclass[conference]{IEEEtran}
\IEEEoverridecommandlockouts
\usepackage{cite}
\usepackage{amsmath,amssymb,amsfonts}
\usepackage{algorithmic}
\usepackage{graphicx}
\usepackage{textcomp}
\usepackage{xcolor}
\usepackage{comment}
\usepackage{caption}
\usepackage{amsthm}

\newtheorem{definition}{Definition}
\newtheorem{proposition}{Proposition}

\def\BibTeX{{\rm B\kern-.05em{\sc i\kern-.025em b}\kern-.08em
    T\kern-.1667em\lower.7ex\hbox{E}\kern-.125emX}}
\begin{document}

\title{Classification of Sequential Circuits \\ as Causal Functions}

\author{\IEEEauthorblockN{Shunji Nishimura}
\IEEEauthorblockA{\textit{Department of Information Engineering} \\
\textit{National Institute of Technology, Oita College}\\
Oita, Japan \\
e-mail: s-nishimura@oita-ct.ac.jp}
}

\maketitle

\begin{abstract}
In sequential circuits, the current output may depend on both past and current inputs. However, certain kinds of sequential circuits do not refer to all of the past inputs to generate the current output; they only refer to a subset of past inputs. This paper investigates which subset of past inputs a sequential circuit refers to, and proposes a new classification of sequential circuits based on this criterion. The conventional classification of sequential circuits distinguishes between synchronous and asynchronous circuits. In contrast, the new classification consolidates synchronous circuits and multiple clock domain circuits into the same category.
\end{abstract}


\section{Introduction}
Digital sequential circuits are typically classified into two categories: synchronous circuits and asynchronous circuits.
A circuit that has memory elements controlled by a clock is classified as synchronous, whereas a circuit without such memory elements is classified as asynchronous.
However, this conventional classification scheme comprising one specific kind and the remainder may not facilitate a thorough understanding of sequential circuits.
For instance, a circuit featuring memory elements with two clock domains, also known as a multiple clock domain circuit, is classified as asynchronous, but it exhibits characteristics similar to those of synchronous circuits rather than those of other typical asynchronous circuits.

To broaden the conventional category of synchronous circuits, one could simply define it as depicted shown in Fig. \ref{fig:SyncCircDef}, where memory control signals $C$ may comprise multiple signals, and multiple clock domain circuits fall within this category.
However, since any sequential circuit could potentially be viewed as a memory element, this definition encompasses all sequential circuits and thus fails to serve as a useful classification.
Therefore, some form of specification for the memory elements may be necessary for this approach to be effective.
\begin{figure}[t]
  \centering
  \includegraphics[width=0.45\linewidth]{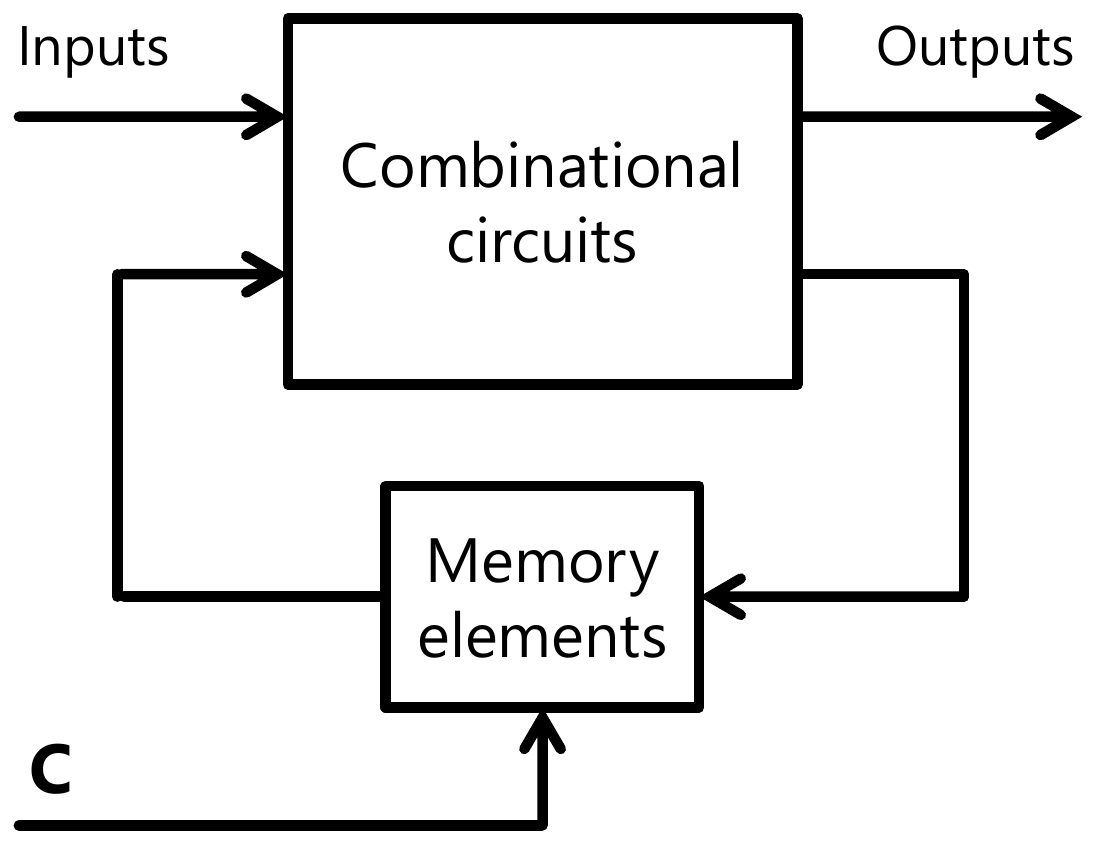}
  \caption{Attempt to define a new class of sequential circuits}
  \label{fig:SyncCircDef}
\end{figure} 

Of course, you could directly define an arbitrary number of D flip-flops allowed for synchronous-like circuits, as employed in \cite{dietmeyer1979logicP306}.
However, in this case, the classification again comprises one specific knd and the remainder, and thus an alternative theoretical or semantical classification may still be desirable.
In the previous study \cite{fujiwara2000new,ooi2005classification}, classifications of synchronous sequential circuits in relation to test generation are presented. However, these classifications are specific to synchronous circuits rather than all sequential circuits.

In this paper, sequential circuits are viewed as causal functions, the outputs of which depend solely on past and current inputs, but not on future inputs. As such, we can use functional expressions to describe sequential circuits, which enable us to examine how a circuit refers to past inputs to generate current output. A specifec manner of referring to past inputs inspires a new classification of sequential circuits.

This study builds upon \cite{BroadSenseSync} and chapter 5 of \cite{MyPhD}, in which an expanded concept of conventional synchronous circuits is presented.
The fundamental idea in these two references and this paper is similar; circuits are viewed as causal functions and we investigate how the functions refer to past inputs. However, the main differences are as follows.
In previous works, the proposed classification aims to extend the class of conventional synchronous circuits.
In contrast, this paper's classification does not, as 
we postulate that all stored states must be updated simultaneously in synchronous circuits.
Furthermore, this paper employs type signatures and dependent types to describe proposed notions, leading to more concise expressions for these.

\section{Preliminaries}

\subsection{Example for an elementary concept}
Consider a partial function defined on the natural numbers that subtracts the second argument from the first:
\begin{align*}
f \;:\; \mathbb{N}\times \mathbb{N} &\to \mathbb{N} \oplus \bot\\
(n,m)\;\; &\mapsto \begin{cases}
    n-m & (n-m \in \mathbb{N}) \\
    \;\;\;\bot & (otherwise)
  \end{cases},
\end{align*}where $\mathbb{N}$, $\oplus$ and $\bot$ are natural numbers $\{0,1,2,\cdots\}$, direct sum and undefined respectively.
The function $f$ is called partial because the actual domain $f^{-1}(\mathbb{N})$ is $\{(n, m) \;|\; n \geq m\}$ and is not equal to the entire $\mathbb{N}\times\mathbb{N}$. For instance, $f(1,2)$ is undefined.
If we try to express $f$ without $\bot$, we can get a type signature
\[f : \{(n, m) \in \mathbb{N}\times\mathbb{N}\;|\; n \geq m\}  \to \mathbb{N},\]
while in this article, we adopt a viewpoint of the first argument restricts the second argument.
Considering $f$ from that viewpoint, first arguments are arbitrary from $\mathbb{N}$ but second arguments are not; they are restricted by first arguments, such as first argument $1$ make second domain $\{0,1\} \subset \mathbb{N}$ and first argument $2$ make second domain $\{0,1,2\}$.
Another type signature of $f$ can be expressed as (in an intuitive expression):
\[f: \mathbb{N}\times\{m\in\mathbb{N}\,|\,\text{\scriptsize m is greater than or equal to the first argument}\} \\ \to\mathbb{N},\]
and that explains the restriction relation between the first and second arguments.
For this example, the restriction relation is exactly ``$\geq$'', and that is all.
Nevertheless, from this viewpoint, we will discover interesting notions about sequential circuits.

\subsection{Dependent types}

We need notions of dependent type \cite{martin1975intuitionistic} on sets and elements instead of types and terms.
For given $A: Set$ and $\alpha: A \to Set$, a subset $(a: A) \times \alpha\,a \,\subset\, A \times {\displaystyle \sum_{x\in A}}\,\alpha\,x$ is defined as
\[(x,y) \in (a: A) \times \alpha\,a \enspace\;\text{iff}\;\enspace x\in A \;\text{and}\; y\in \alpha\,x,\]
where $(a : A)$ represents a domain $A$ and taking arbitrary $a$ from $A$ for later expressions.
Note that in this article, $(a : A)$ denotes a domain with dependent types as explained above, and $a\in A$  denotes the usual proposition.

\subsection{Domain restriction}

With the use of dependent types, we can consider a partial function $A\times B\to O$, of which the second domain $B$ will be restricted by an argument from the first domain $A$. Such a function $f$ could be described with dependent types as
\begin{equation*}
f\,:\, (a: A)\times \alpha\,a\to O,
\label{eq:FirstRestriction}
\end{equation*}where $A,B\in Set$ and $\alpha: A\to\mathcal{P}(B)$, as a power set of $B$.
The function $f$ still matches rough notation $f \subset A \times B \to O$, but at the second domain $B$, $f$ can only take $b \in B$ that is restricted by $\alpha$. Let us call such $\alpha$ a \textbf{restriction map}.
To introduce restriction maps to plural domains, let us modify $\alpha : A \to \mathcal{P}(B)$ to $\alpha' : A \to \mathcal{P}(A\times B)$ of the entire domain (we do not care how $\alpha'$ works for $A$ here), then the above $f$ becomes
\[f\,:\, (a: A)\times \pi_{\small B} (\alpha'\,a) \to O,\]
where $\pi_{\small B} : A \times B \to B$ denotes the projection map of Cartesian products.
Now we can describe domain restrictions for plural products. For given $A,B,C\in Set$, and restriction maps $\alpha: A\to\mathcal{P}(A\times B\times C)$ and $\beta : B\to\mathcal{P}(A\times B\times C)$, a function with two restriction maps become
\begin{equation}
g: (a:A)\times (b: \pi_{\small B} (\alpha\,a))\times \pi_{\small C} (\alpha\,a\cap\beta\,b)\to O,
\label{eq:NoC}
\end{equation}where $\pi_B$ and $\pi_C$ denote projection maps.
Note that $g$ still matches $A \times B \times C \to O$ roughly.
In addition, above (\ref{eq:NoC}) can be described with an unused element $c\in C$ as
\begin{equation}
g\,:\,(a: A)\times (b: \pi_{\small B} (\alpha\,a))\times (c: \pi_{\small C} (\alpha\,a\cap\beta\,b))\to O.
\label{eq:ABCrestriction}
\end{equation}

In a similar way, corresponding to a rough notation ${\displaystyle \prod_{i=0,\cdots ,n}A_i} \to O$, we have a domain restriction expression
\[{\displaystyle \prod_{i=0,\cdots,n}}\left(a_i: \pi_i \left( \bigcap_{j=0,\cdots,i}\alpha_j\,a_j\right)\right) \to O,\]
where $\alpha_{j}: A_j\to\mathcal{P}\left(\displaystyle \prod_{i=0,\cdots ,n}A_i\right)$ are restriction maps for $j=1,\cdots,n$ and $\alpha_{0}$ is a constant map to the entire domain $\displaystyle \prod_{i=0,\cdots ,n}A_i \in \mathcal{P}\left(\displaystyle \prod_{i=0,\cdots ,n}A_i\right)$.

\subsection{Exponential maps}

With regard to exponential maps as in common expression $B^A$, consider them as a part of Cartesian products, i.e., $B^A\in \mathcal{P}(A\times B)$. Then for given
$A,B,C$ : set and $f: A\times C^B \to O$, a domain restriction expression can be
\[(a: A)\times \pi_{\small C}(\alpha\,a)^{\pi_{\small B}(\alpha\,a)}\to O,\]
where $\alpha: A\to\mathcal{P}(A\times B\times C)$ is a restriction map.
In general, substituting $\left(a_i: \pi_i \left( \bigcap_{j=0,\cdots,i}\alpha_j\,a_j\right)\right)$ for each domain$A_i$ of rough notation makes a corresponding restriction expression.

\subsection{Currying and abbreviation}

Considering a restriction expression
\[(a: A)\times (b: \pi_{\small B}(\alpha\,a))\times \pi_{\small C}((\alpha\,a)\cap (\beta\,b))\to O,\]
when the technique of \textit{currying} is employed, it becomes a higher-order function:
\[(a: A)\to (b: \pi_{\small B}(\alpha\,a))\to \pi_{\small C}((\alpha\,a)\cap(\beta\,b))\to O,\]
which means receiving the first argument from $A$ generates a new function that takes its next argument from $B$ (to be exact, $\pi_{\small B}(\alpha\,a)$),  and so on.
It looks better to place restriction maps on the arrows as
\[(a: A)\xrightarrow{\vert_{\alpha\,a}} (b: B)\xrightarrow{\vert_{\beta\,b}} C\to O,\]
with conventional restriction symbol ``$\vert$'' of functions, such as $f\vert_A$.
Since each argument $a, b,\cdots$, is referred to consistently, we can use the further abbreviation as
\[A\xrightarrow{\vert_\alpha} B\xrightarrow{\vert_\beta} C\to O.\]
Concerning the function application, we use the {\it lambda calculus} approach such as $f\,a\,b$ for curried functions, instead of $f(a,b)$ for Cartesian product functions, in the rest of this article.
Finally, these developments are summarized as follows.
\begin{definition} (Domain restriction with abbreviated notation) \label{def-DR}\\
For given restriction maps $\alpha_j : A_j\to \mathcal{P}\left(\displaystyle \prod_{i=0,\cdots ,n}A_i\right) (j=0,\cdots,n)$,
\[A_0\xrightarrow{\vert_{\alpha_0}} A_1\xrightarrow{\vert{\alpha_1}} A_2\xrightarrow{\vert{\alpha_2}}\cdots \xrightarrow{\vert{\alpha_{n-1}}} A_n \xrightarrow{\vert{\alpha_n}} O\]
is an abbreviation for
\begin{eqnarray*}
(a_0: A_0)\to (a_1: \pi_1(\alpha_0\,a_0))\to (a_2: \pi_2(\alpha_0\,a_0 \cap \alpha_1\,a_1))\to \\
\quad\quad\cdots\to (a_n: \pi_n(\alpha_0\,a_0 \cap \alpha_1\,a_1\cap\cdots\cap\alpha_{n-1}\,a_{n-1}))\to O.
\end{eqnarray*}
\end{definition}
Note that the last argument $a_n$ is unnecessary as a matter of fact.

\subsection{Extension of restriction maps}

Considering partial functions of $ A\to B\to C\to O$, if there is a map $\alpha : A \to \mathcal{P}(B)$, we can extend $\alpha$ to $\alpha'$ for entire domain $A\times B\times C$ (in a supremum way) as
\begin{align*}
\alpha' \,:\, A &\to \mathcal{P}(A\times B\times C) \\
\,a &\mapsto A\times(\alpha\,a)\times C,
\end{align*}and use it as a restriction map as
\[A\xrightarrow{\vert{\alpha'}}B\to C\to O.\]
Henceforth, we will not distinguish extended restriction maps from original maps when the entire domain has been clearly given.

\section{Circuits as functions with type signatures}

Sequential circuits can be treated as causal functions as shown below.
For a partially ordered set $T$ as time, a set of input values $I$, and a set of output values $O$, a sequential circuit $f_0$ may be described as
\begin{equation}
f_0\,:\, I^T\to O^T,
\label{eq:CircF0}
\end{equation}where $I^T$ and $O^T$ denote input and output signals respectively.
However this description does not consider {\it causality}. If you want to describe sequential circuits as functions, they must be causal functions, which means the current output depends on only past and current inputs but not on future inputs.
Therefore, we have to consider a causal function expression
\begin{equation}
f\,:\, (t: T)\to (I^{T_{\leq t}} )\to O,
\label{eq:CircF}
\end{equation}where $T_{\leq t} := \{t' \in T\,|\, t' \leq t\}$. In this description, time-given function $f\,t\,$ has a type of $I^{T_{\leq t}}\to O$, i.e., a function that the current output depends on past and current inputs.

When an input signal $\iota \in  I^T$is given, the output signal can be constructed by $f$ as
\begin{align*}
f_0' \,:\, I^T &\to O^T\\
\iota\;\; &\mapsto (t \mapsto f\,t\;(\iota \,\vert{T_{\leq t}}),
\end{align*}
where $\iota \,\vert_{T_{\leq t}}$ denotes a function that is $\iota$, but the domain is restricted to $T_{\leq t} \subset T$.
In the end, causal function $f$ can also be seen as having type signature (\ref{eq:CircF0}), and thus we can consider (\ref{eq:CircF}) as a sequential circuit.

Toward a much simpler expression, if we could read $\leq$ as
\begin{align*}
\leq\;:\, T &\to \mathcal{P}(T)\\
t\; &\mapsto T_{\leq t},
\end{align*}expression (\ref{eq:CircF}) becomes
\begin{equation}
f \,:\, T \,\xrightarrow{\vert_\leq}\, I^T \to O.
\label{eq:SequentialCircuit}
\end{equation}

\subsection{D flip-flops}

\label{ssec:DFF}
As our first example of sequential circuits, we choose the {\it D flip-flops}, which latches an input value when the clock ticks.
To be more precise, the D flip-flop is depicted as shown in Fig. \ref{fig:SymbolDFF}, and it behaves as expressed in TABLE \ref{tbl:DFF}, assuming it is a positive-edge-triggered D flip-flop.
\begin{figure}[b]
  \centering
  \includegraphics[width=25mm]{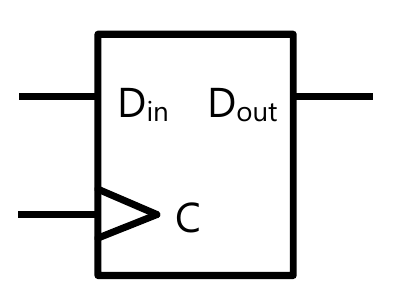}
  \caption{Common symbol of D flip-flops}
  \label{fig:SymbolDFF}
\end{figure}
\begin{table}[b]
  \centering
  \small
  \caption{Truth table of D flip-flops}
  \label{tbl:DFF}
  \begin{tabular}{cc|c} \hline
   $C$ & $D_{in}$ & $D_{out}$ \\
   \hline
   $\uparrow$ & $0$ & $0$ \\
   $\uparrow$ & $1$ & $1$  \\
   otherwise & don't care &  last $D_{out}$ \\
   \hline
   \\
  \end{tabular}
 \end{table}
Input $D_{in}$ is reflected in output $D_{out}$ at the time that clock $C$ changes $0$ to $1$, which is called a positive edge and d enoted by $\uparrow$ in TABLE 
\ref{tbl:DFF}, and otherwise output $D_{out}$ remains.
\begin{figure}
 \centering
  \includegraphics[width=40mm]{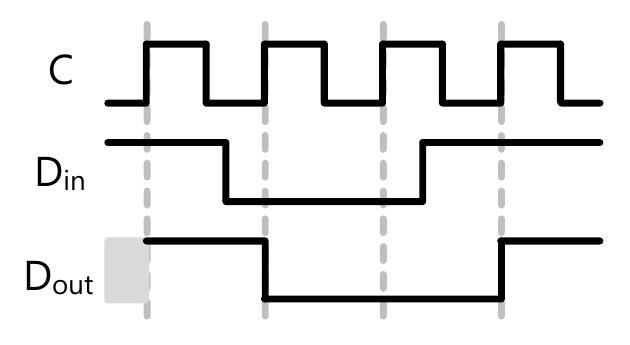}
  \caption{A timing diagram of D flip-flops}
  \label{fig:TimingDiagramDFF}
\end{figure}
An example behavior is depicted in the timing diagram Fig. \ref{fig:TimingDiagramDFF}, and it can be seen that $D_{out}$ simply generates $D_{in}$ value at the last positive edge of $C$.

A causal function $f$ of D flip-flop could be expressed in the same way as type description (\ref{eq:SequentialCircuit}) as
\begin{equation*}
f \,:\, T \,\xrightarrow{\vert_\leq}\, (C\times D_{in})^T \to D_{out},
\label{eq:DFF1}
\end{equation*}and since $(C\times D_{in})^T$ is isomorphic to $C^T\times D_{in}^{\;\,T}$, we can rewrite this as
\[f \,:\, T\,\xrightarrow{\vert_\le}\,C^T\to D_{in}^{\;\, T}\to D_{out}.\]
Taking the first argument $t\in T$, the second and third domains are restricted to $C^{T_{\leq t}}$ and $D_{in}^{\;\,T_{\leq t}}$ respectively. Furthermore, taking the second argument, clock signal $\sigma\in C^{T_{\leq t}}$, we do not need the third argument of signal $D_{in}^{\;\,T_{\leq t}}$ on all the pasts and current, but only need the $D_{in}$ value at the last positive edge of clock signal $\sigma$.
Thus, introducing restriction map $\chi$ as, for a given $t \in T$, 
\begin{align}
\label{eq:ChiOfDFF}
\chi\,:\, C^{T_{\leq t}}&\to\mathcal{P}(T)\\
\sigma\;\;\,&\mapsto \{{\scriptstyle \text{the time of the latest positive edge of}}\;\sigma\}, \nonumber
\end{align}
we finally obtain a precise type signature for a D flip-flop:
\begin{equation}
f \,:\, T\,\xrightarrow{\vert_\le}\,C^T\,\xrightarrow{\vert_\chi}\,D_{in}^{\;\, T}\to D_{out},
\label{eq:DFFfinal}
\end{equation}
and the function body is defined as
\begin{align*}
f\,t\,\sigma \,:\, D_{in}^{\;\,(\chi\,\sigma)} (\cong D_{in}) &\to D_{out}\\
d\quad\quad\quad &\mapsto \;d,
\end{align*}i.e., the identity function.
Expression (\ref{eq:DFFfinal}) indicates that the first map $\leq$ restricts the entire $T$ to the past and current inputs, and the second map $\chi$ also restricts it to only the last positive-edge time.
Note that for a given $t \in T$ and $\sigma \in C^{T_{\leq t}}$, if there is no positive edge in $\sigma$, $\chi\,\sigma$ becomes $\emptyset$, then $f$ is not defined in that situation.

\subsection{SR latch}

An SR(Set Reset) latch has two inputs $S$ to set and $R$ to reset, and output $Q$; it behaves as shown in TABLE \ref{tbl:SRlatch}.
\begin{table}
  \centering
  \small
  \caption{Truth table of SR latch}
  \label{tbl:SRlatch}
  \begin{tabular}{cc|c} \hline
   S & R & Q \\ \hline
   0 & 0 & last Q \\
   0 & 1 & 0  \\
   1 & 0 & 1  \\
   1 & 1 & 0  \\ \hline
  \end{tabular}
\end{table}
This could also be expressed in the same way as type description (\ref{eq:SequentialCircuit}) as
\[T \,\xrightarrow{\leq}\, (S\times R)^T \to Q.\]
However, since any input signal $S^T$ does not restrict the other input signals $R^T$ and vice versa, its type signature cannot develop further with restriction maps, thus restriction maps do not derive any benefit for the type signature of SR latches.
In contrast, in the case of D flip-flops above, clock signals $C^T$ restrict input signals $D_{in}^{\;\, T}$, and in that situation, restriction maps provide an accurate specification on their type signatures.

\section{Concept of time-preserving}

Based on the observations made in the previous section, in order to derive benefits from utilizing type descriptions with restriction maps, we should concentrate on certain types of circuits or causal functions as follows.
\begin{itemize}
    \item circuits that possess clock-like control signals, such as a clock for synchronous circuits
    \item those control signals restrict subsequent domains
\end{itemize}
A formal expression is provided in the next definition.
\begin{definition} Fundamental form \\
For a set of input values $I$, output values $O$, partial ordered set $T$ as time, and set $C$ as a type signature of control signals, a causal function $f$ in which controll signal $\sigma\in C^T$ restricts the latter input $I^T$ in the following manner:
\begin{equation}
f \,:\, T\,\xrightarrow{\vert_\le}\,C^T\,\xrightarrow{\vert_\chi}\,I^T\,\to\, O,
\label{eq:FundamentalForm}
\end{equation}where $\chi\,:\, C^T\to\mathcal{P}(I^T)$ is called fundamental form.
\label{def:FundamentalForm}
\end{definition}
Note that since an actual domain of $\chi$ will be restricted by given $t$, the domain becomes $(t: T)\times (C^{T_{\leq t}})$ to be precise.
We also express a type signature of $\chi$ at given $t$ as $\chi_t\,:\, C^{T_{\leq t}}\to\mathcal{P}(I^{T_{\leq t}})$.


With respect to a circuit $f$ of fundamental form (\ref{eq:FundamentalForm}), we will closely investigate the first two domains $T$ and $C^T$, thereby reverting them from the abbreviated description:
\[f\,:\,(t: T)\times (C^{T_{\leq t}} )\,\xrightarrow{\vert_\chi}\,I^T\,\to\, O.\]
An element $(t,\sigma)$ of the domain $(t: T)\times (C^{T_{\leq t}} )$ is a $C$-valued signal until the current time $t$, which is likened to a clock signal in a conventional notion.
We are introducing a partial order to that domain, naturally derived by the partial order of time $T$.
\begin{definition} Order on causal signals \\
A partial order on $(t: T)\times (C^{T_{\leq t}})$  is defined as:
\[
(t,\sigma)\leq (t',\sigma')\;:\Leftrightarrow\; t\leq t' \;\text{and}\;\, \forall u\leq t,\,\sigma\,u\,=\,\sigma'\,u.
\]
\label{def:TheOrder}
\end{definition}
The last equation of the definition says that $\sigma$ and $\sigma'$ have the same values along with the past of $t$, i.e., their common domain.
By contrast, if there is a time $u\le t$ such that $\sigma\,u\neq \sigma'\,u$, they are not in the order, not in the relation of past and future.
These indicate, in short, ``you cannot change the past,'' and thus, the definition is quite consistent with our conceptual interpretation of time and signals.
In fact, when a signal $\sigma\in C^T$ is given, as is necessary in the real world, each $\sigma$-involved part of $(t: T)\times (C^{T_{\leq t}} )$ at $t$ becomes
\[
\{\,(t,\sigma_t)\,|\,t\in T\,\}\,\cong\,T,
\]
where $\sigma_t $ is a function that is basically $\sigma$ but with the domain restricted to $T_{\leq t}$. That $\sigma$-involved part is order-isomorphic to $T$, and the proposed order appears to be sufficiently convincing.
In addition, the order of Def. \ref{def:TheOrder} becomes {\it prefix order} \cite{cuijpers2013prefix,van2001branching} in terms of transition systems and stream functions.

Thus we have brought a partial order into $(t: T)\times (C^{T_{\leq t}})$ derived from the original time $T$. Next, we consider how the order is to be applied to the final domain $\mathcal{P}(I^T)$ through $\chi$.
From this perspective, we propose our definition as below.
\begin{definition} Time-preserving \\
For a given causal function $f$ typed as
\[
f \,:\, T\,\xrightarrow{\vert_\le}\,C^T\,\xrightarrow{\vert_\chi}\,I^T\,\to\, O,
\]where a restriction map $\chi_t\,:\,C^{T_{\leq t}}\to\mathcal{P}(I^{T_{\leq t}})$ for $t\in T$, $f$ is called time-preserving when the image of $\chi$ has a partial order and $\chi$ becomes order-preserving.
\label{def:TimePreserving}
\end{definition}
That is to say, regarding a time-preserving function $f$, a temporal aspect of $T$, i.e., the order of Def. \ref{def:TheOrder}, is reflected in the image of $\chi$, which is the actual final domain of $f$, i.e., $\chi\,((t: T)\times (C^{T_{\leq t}}))\,\subseteq\, I^T$.
In short, for a time-preserving function, a temporal aspect is preserved throughout its domains.
 
With respect to general relations on the image of $\chi$, we can consider derived order $\chi\, R$ , where $R$ is the partial order on $(t: T)\times (C^{T_{\leq t}})$ of Def. \ref{def:TheOrder} as follows: 
for a given partial order $R \,\subseteq\, ((t: T)\times (C^{T_{\leq t}}))\,\times\,((t: T)\times (C^{T_{\leq t}}))$, the derived relation is determined by
\[
\chi\, R\,=\,\{\,(\chi\,x,\chi\,y)\;|\;(x,y)\in R\,\} \subseteq\,\mathcal{P}(I^T)\times \mathcal{P}(I^T).
\]
The derived order $\chi\, R$ is an important touchstone for judging whether a function is time-preserving.
Indeed, we immediately obtain the next proposition.
\begin{proposition}
For a function $f$ of Def. \ref{def:TimePreserving} and a partial order $R$ on $(t: T)\times (C^{T_{\leq t}})$, if derived relation $\chi\, R\,$ becomes a partial order, $f$ is time-preserving.
\label{prop:OP}
\end{proposition}
\begin{proof}
This naturally follows from the definition of $\chi\, R\,$.
\end{proof}
 In addition, the next proposition will be convenient for determining non-time-preserving functions.
\begin{proposition}
For a function $f$ of Def. \ref{def:TimePreserving} and a partial order $R$ on $(t: T)\times (C^{T_{\leq t}})$, if derived relation $\chi\, R\,$ does not become partial order, $f$ is not time-preserving.
\label{prop:nonOP}
\end{proposition}
\begin{proof}
By contradiction: let us assume binary relation $\chi\, R$ is not a partial order, but that $\chi$ is order-preserving with a partial order $R'$ on the image of $\chi$.
Since $\chi\,R$ is a relation derived from $R$, reflexivity and transitivity are satisfied naturally,
with the consequence that we can focus on antisymmetry.
Let us take $s_0,s_1,u_0,u_1\in (t: T)\times (C^{T_{\leq t}})$ s.t.
\[
s_0\le s_1,\, u_0\le u_1,\, \chi\, s_0 = \chi\, u_1,\, \chi\, s_1 = \chi\, u_0,\;\text{and}\; \chi\, s_0 \neq \chi\, s_1.
\]
Assuming that $\chi\,s_0 \le_{R'} \chi\,s_1$, $\chi\,s_1 \le_{R'} \chi\,s_0$ never holds since $R'$ is a partial order, 
even though the fact that $u_0\le u_1$ and $\chi$ is order-preserving yield ($\chi\,u_0=$)$\chi\,s_1 \le_{R'} \chi\,s_0(=\chi\,u_1)$. These provide a contradiction.
\end{proof}

\section{Classification of sequential circuits}

\subsection{Time-preserving circuits}

\subsubsection{Synchronous circuits}
as shown in Fig. \ref{fig:DFFsync}, where D-FF denotes D flip-flop, become time preserving as follows.
\begin{figure}[b]
 \centering
 \includegraphics[width=45mm]{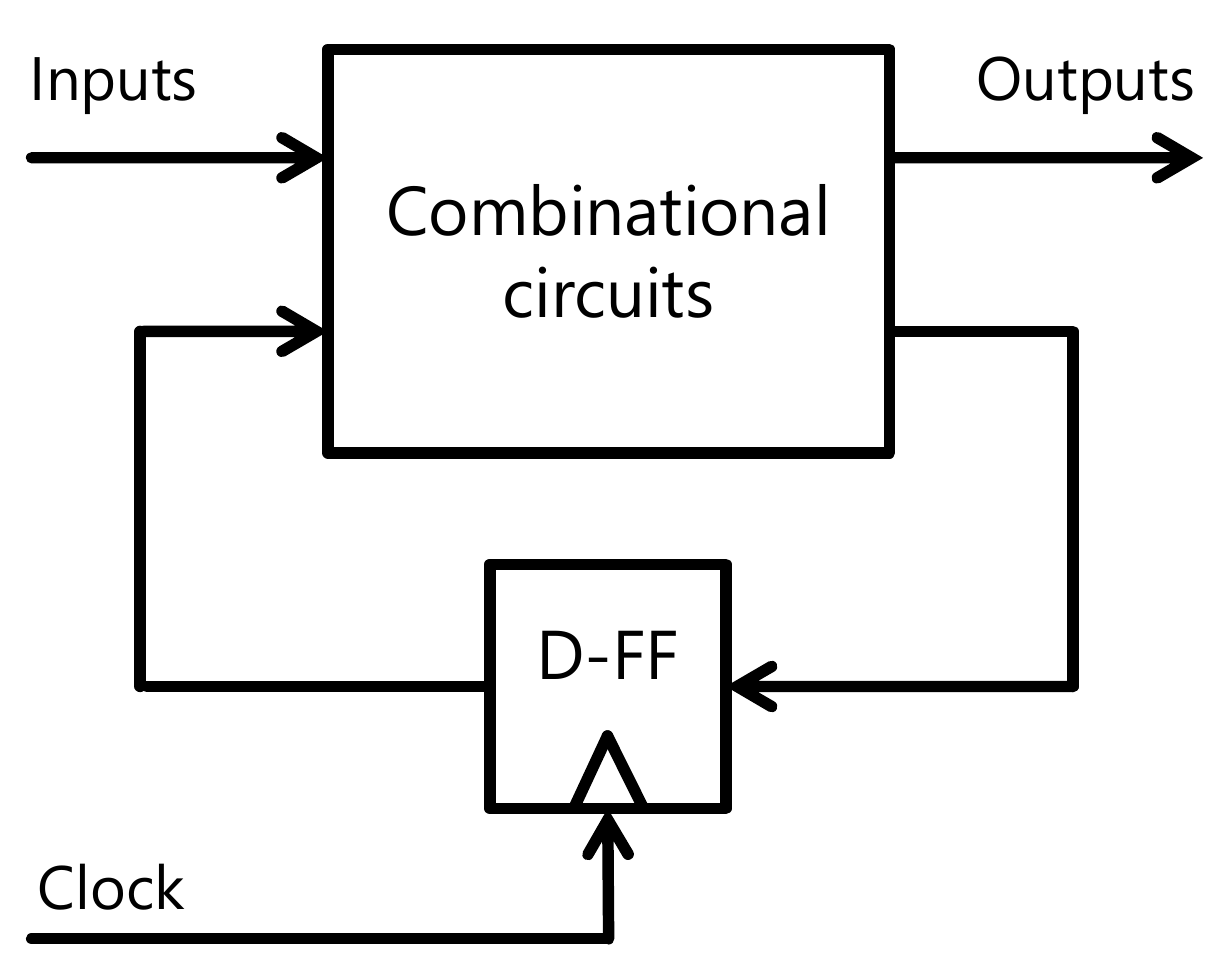}
  \caption{Synchronous circuits}
 \label{fig:DFFsync}
\end{figure}
Let $f$ be a causal stream function of an arbitrary synchronous circuit, then $f$ has the same type signature as (\ref{eq:DFFfinal}).
We are interested in whether the image of $\chi$ preserves temporal aspect of $T$.
The imge of $\chi$ is as shown in Fig. \ref{fig:RefSync}, where $\sigma_{t_n}$ has type signature  $C^{T_{\le t_n}}$.
\begin{figure}
 \centering
 \includegraphics[width=50mm]{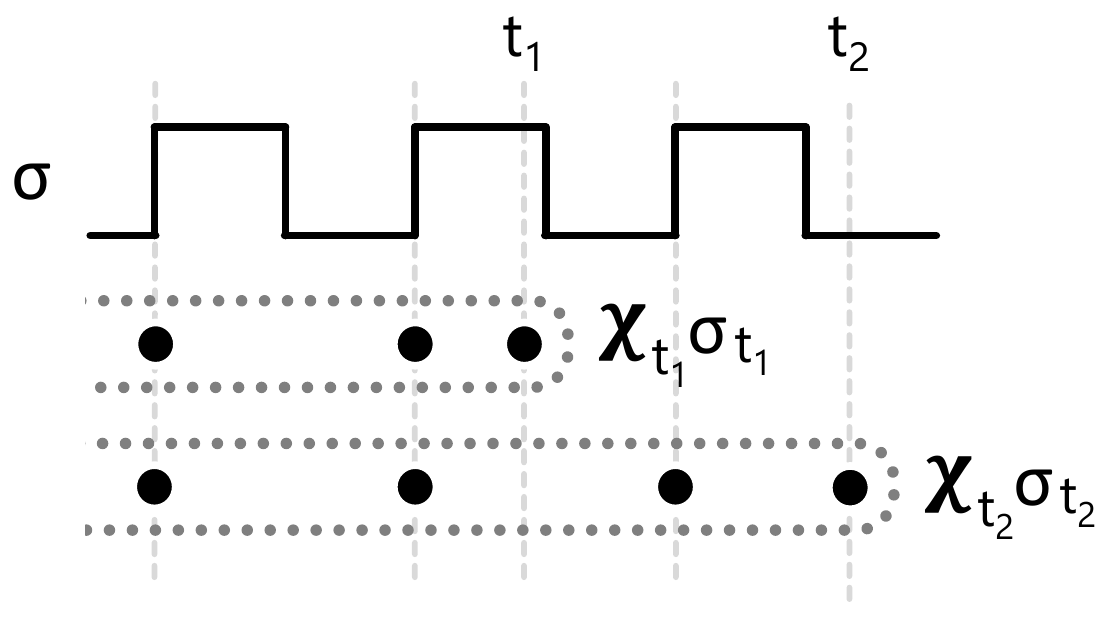}
  \caption{$\chi$-image of synchronous circuits}
 \label{fig:RefSync}
\end{figure}
Since $D_{out}$ of D-FF will be included in $D_{in}$ at the next clock tick, $\chi$ must be a set of all time points of positive edges in the past. 
That $\chi_{t_1}\,\sigma_{t_1}$, the image of $\chi$ at $t_1$ with given control signal $\sigma$, contains a data input of current time $t_1$ and the past data inputs at each positive edges.

Figure \ref{fig:RefSync} indicates that $t_1 \le t_2$ implies (past part of $\chi_{t_1}\, \sigma_{t_1}\,$) $\subseteq$ (past part of $\chi_{t_2}\, \sigma_{t_2}$) and of course (current time $t_1$ of $\chi_{t_1}\,\sigma_{t_1}$) $\leq$ (current time $t_2$ of $\chi_{t_2}\,\sigma_{t_2}$), i.e. both relations are in partial order.
Therefore, by Prop. \ref{prop:OP}, $\chi$ is order preserving and synchronous circuits are time-preserving.

\subsubsection{Multiple clock domain circuits}
Concerning circuits with multiple clock domains as shown in Fig. \ref{fig:MultiClock}, the image of $\phi$ is as shown in Fig. \ref{fig:RefMulti}, where $\sigma = \sigma' \times \sigma'' : (C')^T \times (C'')^T$, it configures inclusion order on the past parts, and it is also time preserving.
\begin{figure}
 \centering
 \includegraphics[width=60mm]{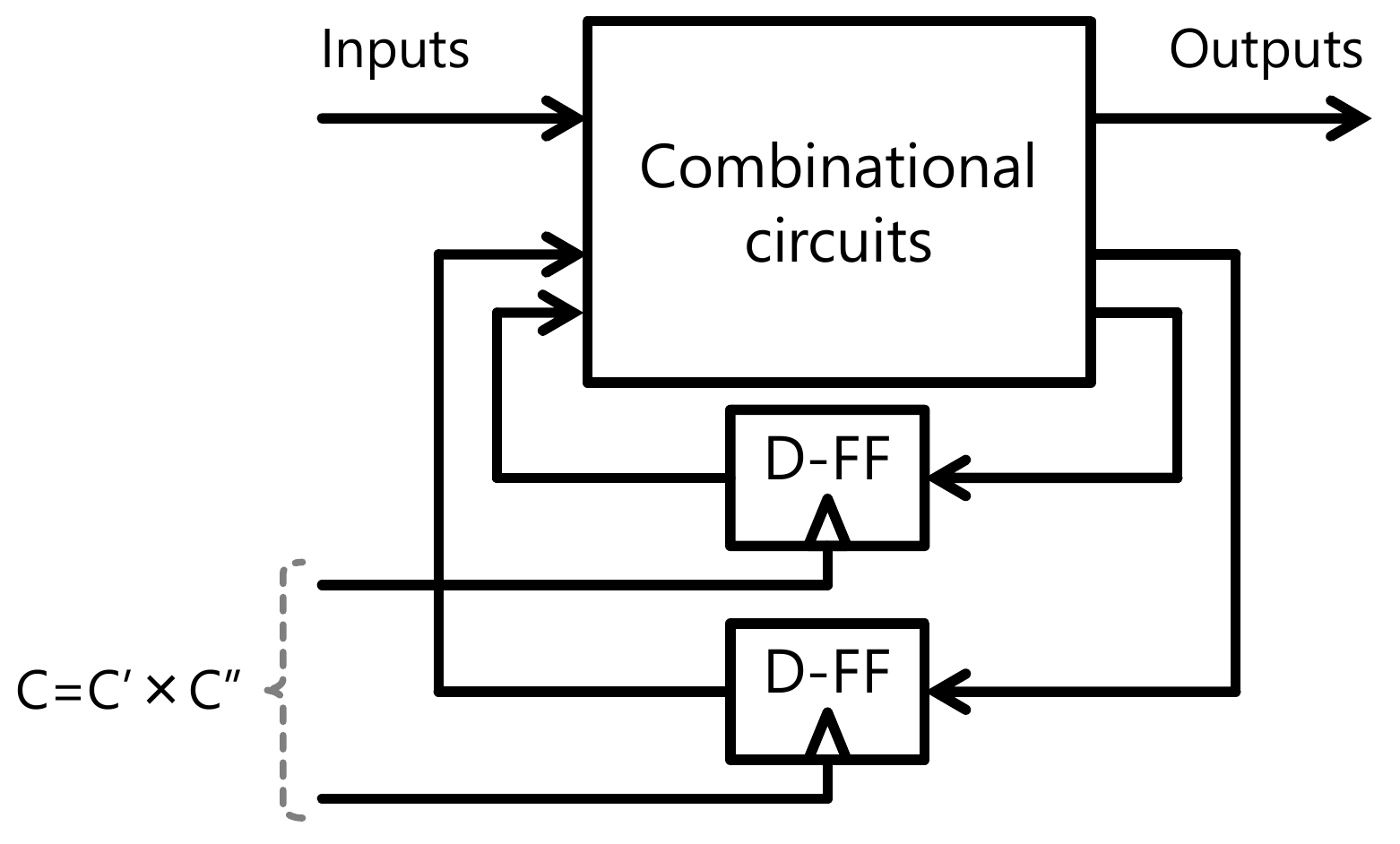}
  \caption{Multi-clock domain circuits}
 \label{fig:MultiClock}
\end{figure}
\begin{figure}
 \centering
 \includegraphics[width=50mm]{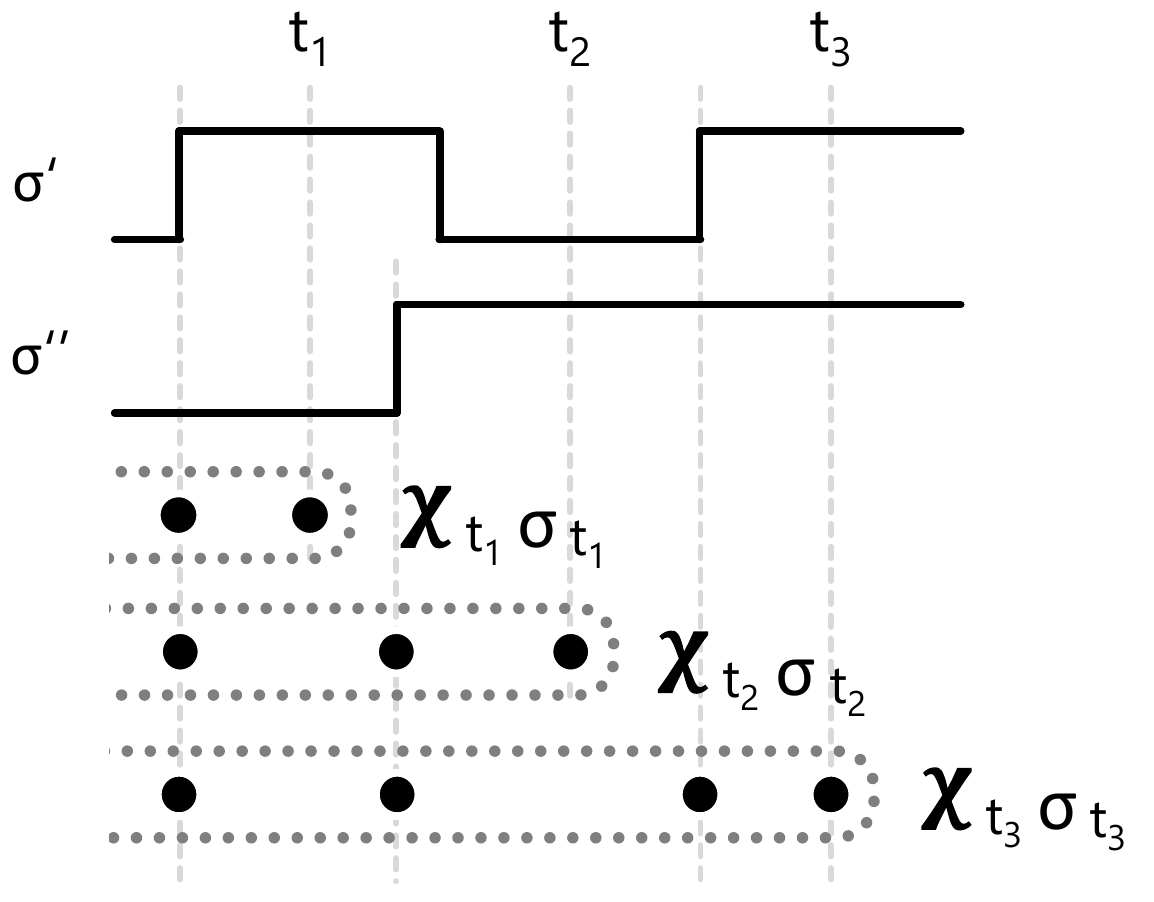}
  \caption{$\chi$-image of multi-clock domain circuits}
 \label{fig:RefMulti}
\end{figure}
However, note that a general expression of multi-clock domain circuits as shown in Fig. \ref{fig:MultiClock} does not adequately describe their essence. For a practical multi-clock domain circuit, its inputs and state inputs/outputs must be mostly divided into each domain.

\subsubsection{Multiplexers}
are not sequential but combinational circuits.
It is a circuit element that selects input signals according to a select signal,  depicted as shown in Fig.\ref{fig:SymbolMUX}.
\begin{figure}[hbt]
 \centering
 \includegraphics[width=15mm]{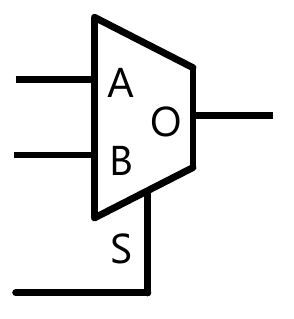}
 \captionsetup{justification=centering}
 \caption{Common symbol of multiplexers}
 \label{fig:SymbolMUX}
\end{figure}
A type signature for the circuit could be described as:
\[
(s : S) \to \chi\, s \to O \\
\]where $S=\{a,b\}$ and
\begin{align*}
\chi\,:\, S&\;\to\; \mathcal{P}(A\times B)\\
a &\;\mapsto\; A\times \emptyset\\
b &\;\mapsto\; \emptyset\times B.
\end{align*}In this example, filtering map $\chi$ is applied to data inputs $A\times B$, and it is also time-preserving because the image of $\chi$ configures inclusion order.

\subsection{Besides time-preserving circuits}

Figure \ref{fig:BlockDiagramMem} shows a circuit with memory that has two addresses $A$ and $B$, these are able to both read and write.
For two control signals $\sigma, \sigma' : C^T$ shown in Fig. \ref{fig:RefersMem}, this indicates that $\chi_{t_1}\, \sigma_{t_1}\le\chi_{t_2}\, \sigma_{t_2}$ and $\chi_{t_1}\, \sigma_{t_1} (=\chi_{t_2}\, \sigma'_{t_2})\ge\chi_{t_2}\, \sigma_{t_2}(=\chi_{t_1}\, \sigma'_{t_1})$ but $\chi_{t_1}\, \sigma_{t_1}\neq\chi_{t_2}\, \sigma_{t_2}$, thus antisymmetry does not satisfied.
Therefore, by Prop. \ref{prop:nonOP}, the circuit is not time-preserving.
\begin{figure}
 \centering
 \includegraphics[width=70mm]{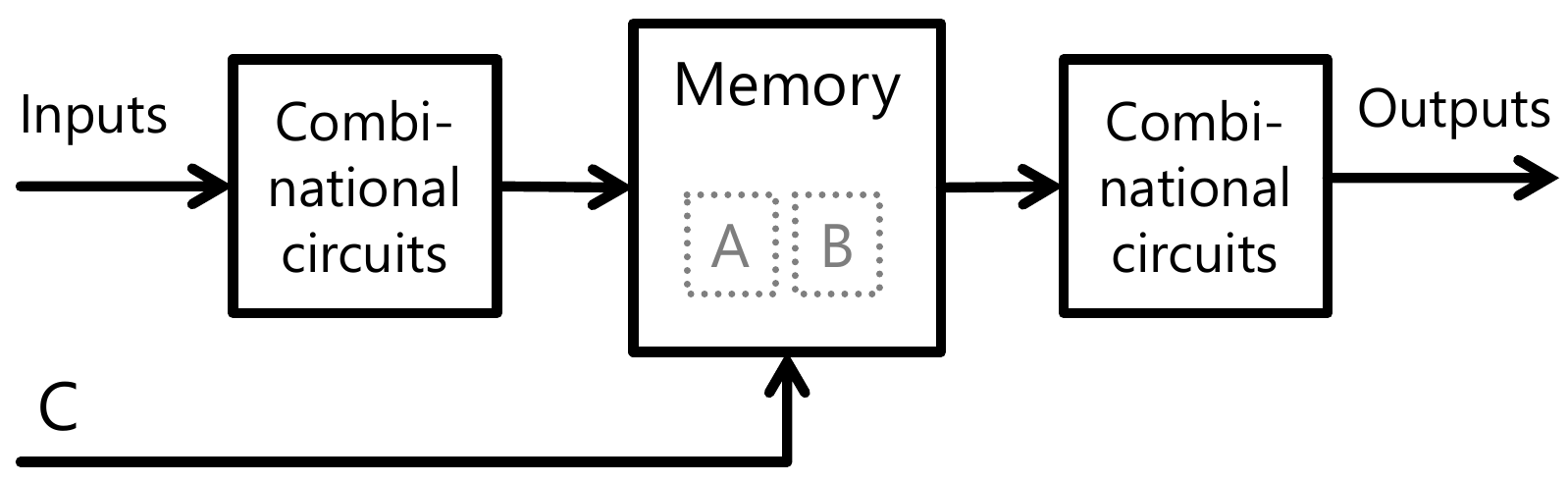}
 \captionsetup{justification=centering}
 \caption{A circuit with A/B memory}
 \label{fig:BlockDiagramMem}
\end{figure}
\begin{figure}
 \centering
 \includegraphics[width=55mm]{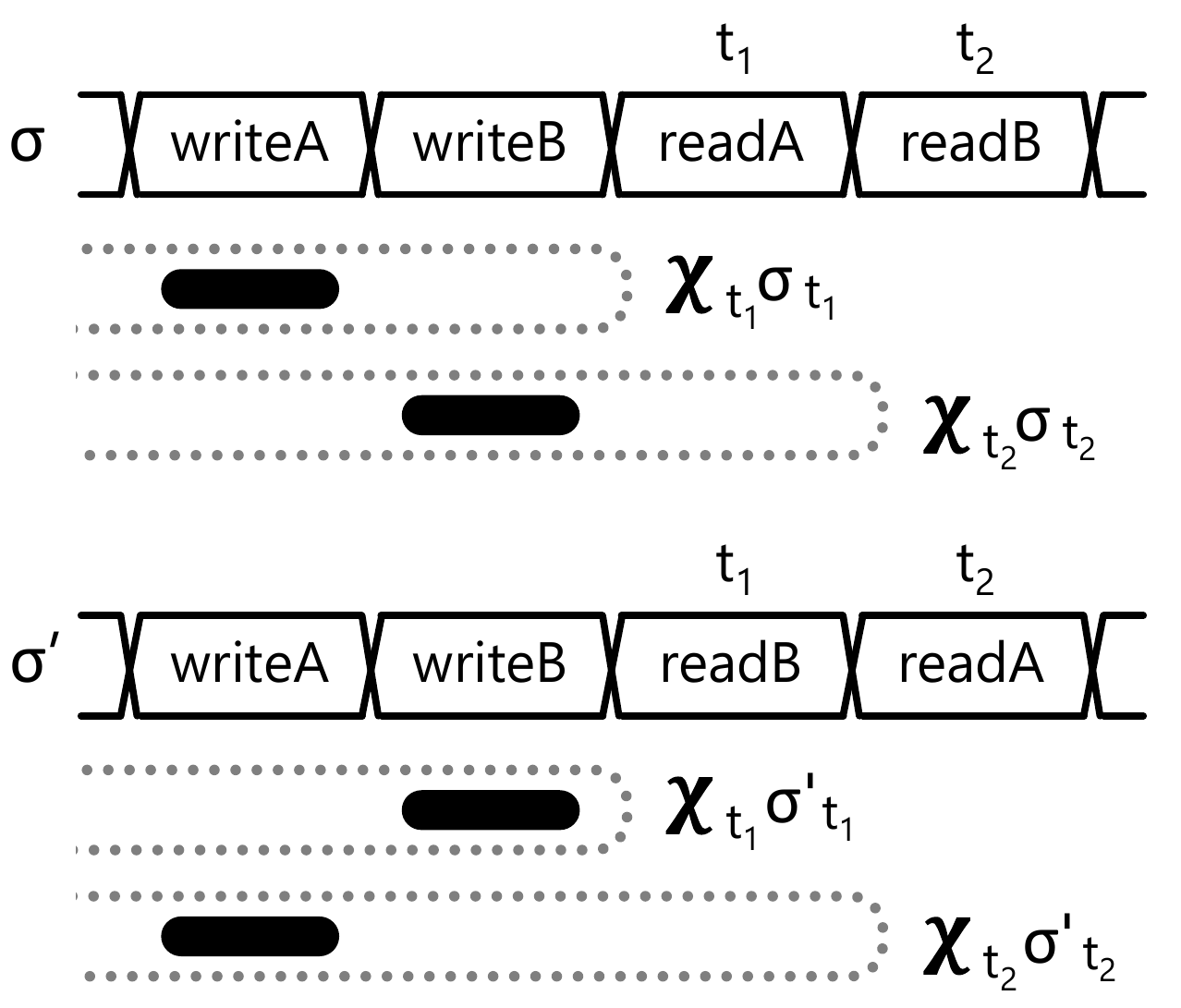}
 \captionsetup{justification=centering}
 \caption{$\chi$-image of the circuit of Fig. \ref{fig:BlockDiagramMem}}
 \label{fig:RefersMem}
\end{figure}

\subsection{Summary of the classification}

In this section, we examined the time preservation of several circuits/causal functions, and as a result, the following well-known circuits are found to be time-preserving:
\begin{itemize}
\item synchronous circuits
\item multi-clock domain circuits.
\end{itemize}
There is a circuit that can be expressed in fundamental form (\ref{eq:FundamentalForm}), but is not time-preserving: the circuit with A/B memory mentioned above.
Typical asynchronous circuits, such as SR latches, do not seem to be able to describe in the fundamental form.

\section{Conclusion}

In this paper, sequential circuits are regarded as causal functions, which is a formal representation of the natural language definition as ``the current output depends on past and current inputs.''
In the expressions of causal functions, we introduce a particular form called fundamental form, in which functions have signals that restrict the domains required to generate the current output.
We conducted a study on the methods of domain restriction and specifically analyzed instances in which a restriction maintains the temporal aspect of the circuit's environment, referred to as time-preserving.

The concept of time-preserving is used for the classification of sequential circuits in the latter half of this paper.
Whereas sequential circuits were previously classified as either synchronous or others (asynchronous), a new classification introduces a broader category encompassing synchronous circuits.
For instance, synchronous circuits and multiple clock domain circuits, which have multiple D flip-flops with distinct clock domains, are classified in the same category from this perspective.


\bibliographystyle{ieeetr}
\bibliography{the}

\end{document}